\documentclass[12pt]{article}

\usepackage{amsmath}
\usepackage{mathptmx}
\usepackage{helvet}
\usepackage{courier}
\usepackage{amssymb}
\usepackage{amsthm}
\usepackage[T1]{fontenc}
\usepackage[latin9]{inputenc}
\usepackage{makeidx}         
\usepackage{graphicx}        
\usepackage{multicol}        
\usepackage[bottom]{footmisc}

\newtheorem{defn}{Definition}
\newtheorem{thm}{Theorem}
\usepackage{hyperref}
\hypersetup{colorlinks=true,
  linkcolor=blue,
  anchorcolor=blue,
  citecolor=red,
  urlcolor=brown,
  pdfauthor={Michele Castellana}}

\providecommand{\nn}{\nonumber}
\providecommand{\be}{\begin{equation}}
  \providecommand{\ee}{\end{equation}}
\providecommand{\bea}{\begin{eqnarray}}
  \providecommand{\eea}{\end{eqnarray}}
\providecommand{\beas}{\begin{eqnarray*}}
  \providecommand{\eeas}{\end{eqnarray*}}

\providecommand{\beni}{\begin{equation*}}
  \providecommand{\eeni}{\end{equation*}}

\providecommand{\bw}{\begin{widetext}}
  \providecommand{\ew}{\end{widetext}}

\makeindex

\date{}

\begin{document}
\author{Michele Castellana\footnote{Lewis-Sigler Institute for Integrative Genomics, Princeton University, Princeton, New Jersey 08544, United States.}, Adriano Barra\footnote{Dipartimento di Fisica, Sapienza Universit\`{a} di Roma, Roma, Italy.}, Francesco Guerra\footnote{Dipartimento di Fisica, Sapienza Universit\`{a} di Roma and INFN Sezione di Roma, Roma, Italy.}}
\title{Free-energy bounds for hierarchical spin models}
\maketitle

\begin{abstract}
  In this paper we study two non-mean-field spin models built on a hierarchical lattice: The hierarchical Edward-Anderson model (HEA) of a spin glass,  and Dyson's hierarchical model (DHM) of a ferromagnet. For the HEA, we prove the existence of the thermodynamic limit of the free energy and the replica-symmetry-breaking (RSB) free-energy bounds previously derived for the Sherrington-Kirkpatrick model of a spin glass.
  These RSB mean-field bounds are exact only if the order-parameter fluctuations (OPF) vanish: Given that such fluctuations are not negligible in non-mean-field models, we develop a novel strategy to tackle part of OPF in hierarchical models. The method is based on absorbing part of OPF of a block of spins into an effective Hamiltonian of the underlying spin blocks. We illustrate this method for DHM and show that, compared to the mean-field bound for the free energy, it provides a tighter non-mean-field bound, with a critical temperature closer to the exact one. To extend this method to the HEA model, a suitable generalization of Griffith's correlation inequalities for Ising ferromagnets is needed: Since correlation inequalities for spin glasses are still an open topic, we leave the extension of this method to hierarchical spin glasses as a future perspective.
\end{abstract}

\section{Introduction}

The mean-field (MF) picture of spin glasses has been extensively studied in the last few decades, and it is now mostly understood at a rigorous level \cite{panchenko2013parisi}.  In particular, the replica-symmetry-breaking (RSB) free-energy picture originally proposed by Parisi  \cite{parisi1983order} for the MF Sherrington-Kirkpatrick (SK) model has been proved to be a rigorous upper bound for the SK free energy in \cite{guerra2003broken}.  Later on, this bound has been shown to be exact in the thermodynamic limit \cite{talagrand2006parisi}. Despite the remarkable progress in understanding the MF picture, the non-mean-field (NMF) scenario of spin glasses is still a  source of debate \cite{young2006numerical}.\\

Among the NMF models of spin glasses, the hierarchical Edward-Anderson model (HEA) has attracted particular interest in recent years \cite{franz2009overlap,castellana2011real,angelini2013ensemble}. The HEA is natural extension of a NMF model of a ferromagnet, Dyson's hierarchical model  (DHM) \cite{dyson1969existence}. In DHM, the ferromagnetic spin-spin couplings are disposed in a hierarchical way: This arrangement  of the couplings allows for a recursive structure which makes DHM particularly suitable for the implementation of renormalization-group methods \cite{bleher1973investigation}. The HEA shares with DHM this hierarchical coupling structure, but it differs from DHM in the nature of the couplings: While DHM has only ferromagnetic--i.e. positive--couplings, in the HEA spin-spin couplings are random variables taking both positive and negative values, thus implying frustration.\\

In this paper, we provide rigorous free-energy bounds for DHM and HEA. For the HEA, we frist prove the existence of the thermodynamic limit of the free energy and its self-averaging property, and then we extend the RSB bound for the MF SK model to the HEA. Given that this MF bound is exact only if the order-parameter fluctuations (OPF) vanish, we provide a new scheme that leverages the hierarchical structure of the model to account for OPF,  thus improving upon the MF bound. In this new scheme, OPF of a hierarchical spin block are absorbed into an effective Hamiltonian of the underlying blocks. We explicitly test this idea for DHM and show that, compared to the MF bound, this new scheme provides  a tighter NMF bound. As a consequence, the NMF-bound critical temperature is closer to the exact value compared to that of the MF bound \cite{griffiths1967correlationsIII}. \\

Given that the proof of the NMF bound for DHM makes use of well-known correlations inequalities for ferromagnetic systems \cite{griffiths1967correlationsII}, to generalize this method to the HEA a suitable generalization of the correlation inequalities to spin glasses is needed. We leave this correlation-inequality extension as a topic of future research \cite{morita2004griffiths,contucci2007correlationI}. If extended to the HEA, our method could provide a novel NMF  bound for the free energy, providing a novel guidance in understanding the low-temperature features of NMF spin glasses.

\section{Hierarchical Edwards-Anderson Model}

The HEA model is a system of $2^{k+1}$ Ising spins $S_i = \pm 1$ labeled by index $i=1,2,\cdots, 2^{k+1}$, whose Hamiltonian $H_{k+1}[\vec{S}]$ is introduced recursively by the following

\begin{defn}
  \label{def1} The Hamiltonian of the hierarchical Edwards-Anderson model (HEA) is defined by
  \be\label{eq10}
  H_{k+1}[\vec{S}]=H_{k}^{1}[\vec{S}_{1}]+H_{k}^{2}[\vec{S}_{2}]-\frac{1}{2^{(k+1)\sigma}}\sum_{i<j=1}^{2^{k+1}}J_{ij}S_{i}S_{j},
  \ee
where $\vec{S}_1 \equiv \{S_i\}_{1 \leq i \leq 2^k}$,  $\vec{S}_2 \equiv \{S_i\}_{2^k+1 \leq i \leq 2^{k+1}}$,  $H_{0}[S]=0$, $J_{ij}$ are  independent and identically distributed (IID)  Gaussian random variables with zero mean and unit variance, and $\sigma$ is a number.
\end{defn}

It is important to point out that the number $\sigma$ in Definition \ref{def1} determines how fast the spin-spin interactions decrease with distance: The larger $\sigma$, the faster the interactions decrease.\\

Let us now prove the existence of the thermodynamic limit for the quenched free energy of the HEA. By using the scheme proposed in  \cite{guerra2002thermodynamic} in a recursive way adapted to the hierarchical structure of the model, we obtain the following

\begin{thm}\label{thm1}
  If $\sigma>1/2$, given a Gaussian random variable $h$ and $2^{k+1}$ IID copies $h_{1},\ldots,h_{2^{k+1}}$ of $h$, let us introduce the free energy
  \[
  f_{k+1}\equiv\frac{1}{2^{k+1}}\mathbb{E}\left[\log\sum_{\vec{S}}\exp\left(-\beta H_{k+1}[\vec{S}]+\sum_{i=1}^{2^{k+1}}h_{i}S_{i}\right)\right],
  \]
  where the inverse-temperature $\beta$ is a non-negative number, and $\mathbb{E}[]$ denotes the expectation with respect to all random variables.
  \newline
  Then, $f \equiv \lim_{k\rightarrow\infty}f_{k+1}$ exists.
\end{thm}
\begin{proof}
  Consider an interpolating parameter $0 \leq t \leq 1$ and the Hamiltonian
  \be \label{eq27}
  H_{k+1,t}[\vec{S}] \equiv - \frac{\sqrt{t}}{2^{(k+1)\sigma}}\sum_{i<j=1}^{2^{k+1}} J_{ij} S_i S_j +  H_{k}^{1}[\vec{S}_{1}]+H_{k}^{2}[\vec{S}_{2}].
  \ee
  The  partition function and  free energy related to the Hamiltonian (\ref{eq27}) are
  \bea
  Z_{k+1,t} & \equiv & \sum_{\vec{S}} \exp\left(-\beta  H_{k+1,t}[\vec{S}] + \sum_{i=1}^{2^{k+1}} h_i S_i \right),\\
  \phi_{k+1,t} & \equiv & \frac{1}{2^{k+1}} \mathbb{E}[ \log Z_{k+1,t}].
  \eea
  For $t=1$, $\phi_{k+1,t}$ equals the free energy of the original model
  \be \label{eq1}
  \phi_{k+1,1} = f_{k+1},
  \ee
  while for $t=0$, $\phi_{k+1,t}$ is given by the free energy of two independent HEAs with $2^k$ spins: By using Definition \ref{def1} for the HEA Hamiltonian, this is exactly $f_k$:
  \be\label{eq2}
  \phi_{k+1,0} = f_k.
  \ee

  To interpolate between $\phi_{k+1,1}$ and $\phi_{k+1,0}$, we compute the derivative of $\phi_{k+1,t}$ with respect to $t$. By integrating by parts over the Gaussian variables $J_{ij}$, it is easy to show that
  \bea\label{eq3}
  \frac{d \phi_{k+1,t} }{dt} &=& \frac{\beta}{2 \sqrt{t} 2^{(k+1)(1+ \sigma)}} \sum_{i<j=1}^{2^{k+1}} \mathbb{E}[J_{ij} \Omega( S_i S_j)_t] \\ \nn
  & = & \frac{\beta^2}{4} 2^{(k+1)(1-2\sigma)}(1- \mathbb{E}[\Omega( R_{12}^2 )_t]),
  \eea
  where $R_{12} \equiv \frac{1}{2^{k+1}} \sum_{i=1}^{2^{k+1}} S^1_i S^2_i$ is the overlap between two independent replicas $\vec{S}^1$, $\vec{S}^2$ and $\Omega$ is the Boltzmann average over the two replicas
  \be
  \Omega( \cdot ) \equiv \frac{ \sum_{\vec{S}^1 \vec{S}^2} \exp\left[ -(H_{k+1,t}[\vec{S}^1]+H_{k+1,t}[\vec{S}^2]) + \sum_{i=1}^{2^{k+1}} h_i (S^1_i+S^2_i) \right]}{Z_{k+1,t}^2}.
  \ee
  From Eq. (\ref{eq3}) we obtain an upper and a lower bound for the derivative of $\phi_{k+1,t}$
  \be \label{eq7}
  0 \leq  \frac{d \phi_{k+1,t} }{dt}  \leq \frac{\beta^2}{4} 2^{(k+1)(1-2\sigma)}.
  \ee

  Putting together Eqs. (\ref{eq1}), (\ref{eq2}) and the upper bound in  Eq.  (\ref{eq7}) we obtain
  \be\label{eq4}
  f_{k+1} = f_k + \int_0^1 \frac{d \phi_{k+1,t} }{dt}  dt   \leq  f_k+  \frac{\beta^2}{4} 2^{(k+1)(1-2\sigma)},
  \ee
  while the lower bound in Eq. (\ref{eq7}) implies
  \be\label{eq8}
  f_{k+1} \geq f_k.
  \ee
  We can now use the recursive structure of HEA to establish the final result: Following a method originally used for the ferromagnetic version of the HEA \cite{dyson1969existence}, we iterate Eq. (\ref{eq4}) for $k+1, k, k-1, \cdots, 0$. As we reach  $k=0$, we are left with the free energy of a one-spin HEA that we can  compute explicitly
  \bea\label{eq5}
  f_{k+1} &\leq& \frac{\beta^2}{4} 2^{(k+1)(1-2\sigma)} + \frac{\beta^2}{4} 2^{k(1-2\sigma)} + f_{k-1} \\ \nn
  & \leq & \cdots\\ \nn
  & \leq & \frac{\beta^2}{4} \sum_{l=1}^{k+1} 2^{l(1-2\sigma)} + \mathbb{E}[\log 2 \cosh(h)].
  \eea
  Since here $\sigma > 1/2$, we have
  \be \label{eq6}
  \sum_{l=1}^{k+1} 2^{l(1-2\sigma)}  \leq \frac{1}{1-2^{1-2\sigma}} < \infty\;\;\;\;\; \forall k\geq 0.
  \ee
  Putting together Eqs. (\ref{eq5}), (\ref{eq6}) we obtain that the sequence $k \rightarrow f_{k+1}$ is bounded  above
  \be
  f_{k+1} \leq \frac{\beta^2}{4}  \frac{1}{1-2^{1-2\sigma}} +\mathbb{E}[\log 2 \cosh(h)] <  \infty \;\;\; \;\; \forall k \geq 0,
  \ee
  and from Eq. (\ref{eq8}) we have that the sequence $k \rightarrow f_{k+1}$ is non decreasing, implying that $\lim_{k\rightarrow \infty} f_{k+1}$ exists.
\end{proof}

Based on previous results on the Sherrington-Kirkpatrick model, it is also easy to show that the free energy of the HEA is self-averaging in the thermodynamic limit

\begin{thm}\label{thm2}
  For $\sigma > 1/2$, the free energy of the HEA is self-averaging in the thermodynamic limit
  \be
  \lim _{k \rightarrow \infty} \frac{1}{2^{k+1}}\log\sum_{\vec{S}}\exp\left(-\beta H_{k+1}[\vec{S}]+\sum_{i=1}^{2^{k+1}}h_{i}S_{i}\right)  = f , \; \;  \textrm{ with probability } 1 .
  \ee
\end{thm}
Theorem \ref{thm2} can be proven by a step-by-step repetition of the proof of free-energy self-averaging for the SK model  \cite{guerra2002thermodynamic,bolthausen2007spin}. \\

We will  now establish a bound for the free energy of the HEA. We start by proving a MF bound for the free energy based on an extension of the RSB free-energy bounds for the SK model \cite{guerra2003broken} by the following

\begin{thm}[Mean-field bound]\label{thm3}
  Consider  $0 \equiv q_{0} \leq q_1  \leq \cdots \leq  q_K \equiv 1$, $0\equiv m_0  < m_1 \leq m_2 \leq \cdots \leq m_K \leq m_{K+1} \equiv 1$ and $K$ IID random variables $z_1, \cdots, z_K$ with zero mean and unit variance. Consider the sequence $Z_0, Z_1, \cdots, Z_K$ defined recursively by
  \bea\label{eq26}
  Z_K &\equiv& \cosh\left( h + \beta \sqrt{\sum_{l=1}^{k+1} 2^{l(1-2 \sigma)}} \sum_{a=1}^{K} \sqrt{q_a - q_{a-1}} z_a \right),\\ \nn
  Z_a^{m_{a+1}} & = & \mathbb{E}_{a+1}\left[ Z_{a+1}^{m_{a+1}} \right],
  \eea
  where $\mathbb{E}_{a}$ denotes the expectation with respect to $z_a$. Then,
  \bea\label{eq38}
  f_{k+1}& \leq & \log 2 + \mathbb{E}[\log Z_0] +  \frac{\beta^2}{4} \sum_{l=1}^{k+1} 2^{l(1-2\sigma)} \left[ \sum_{a=1}^{K} (m_{a+1}-m_a) q_a^2  - 1\right] .
  \eea
\end{thm}
\begin{proof}
  The proof makes use of the RSB bounds for the SK model \cite{guerra2003broken} in a recursive way, suitably adapted to the hierarchical structure of the model. \\

  Let us introduce the interpolating Hamiltonian
  \bea\label{eq25}
  H_{k+1,t}[\vec{S}] &\equiv & - \frac{ \sqrt{t}}{2^{(k+1)\sigma}}\sum_{i>j=1}^{2^{k+1}} J_{ij} S_i S_j +  \sqrt{1-t} 2^{(k+1)(1/2-\sigma)} \times \\ \nn
  && \times \sum_{a=1}^K \sqrt{q_a - q_{a-1}}  \sum_{i=1}^{2^{k+1}} J^{k+1}_{a,i}S_i  +  H_{k}^{1}[\vec{S}_{1}]+H_{k}^{2}[\vec{S}_{2}] ,
  \eea
  where $\{J^{k+1}_{a,i}\}$ are IID Gaussian random variables with zero mean and unit variance.  We introduce the partition functions $Z^{k+1}_{0,t}(h, \{ h' \}), \cdots, Z^{k+1}_{K,t}(h, \{ h' \})$ defined recursively by
  \bea \label{eq22}
  Z^{k+1}_{K,t}(h, \{ h' \}) & \equiv & \sum_{\vec{S}} \exp\left[ - \beta H_{k+1,t}[\vec{S}] + \sum_{i=1}^{2^{k+1}} \left(h_i + \sum_{a=1}^K h'_{a,i}\right)S_i\right],\\ \nn
  Z^{k+1}_{a,t}(h, \{ h' \})^{m_{a+1}} &=& \mathbb{E}_{a+1}\left[ Z^{k+1}_{a+1,t}(h, \{ h' \})^{m_{a+1}}\right],
  \eea
  where $\{ h'_{a,i}\}$  are IID Gaussian random variables and $ \mathbb{E}_{a}[ \cdot ]$ denotes the average with respect to all variables labeled by index $a$. The free energy associated with the Hamiltonian (\ref{eq25}) is
  \be\label{eq9}
  \phi_{k+1,t}(h, \{ h' \}) \equiv \frac{1}{2^{k+1}} \mathbb{E} [\log Z^{k+1,t}_0(h, \{ h' \}) ],
  \ee
  where in the left-hand side (LHS) of Eq. (\ref{eq9}) the dependence of  $\phi_{k+1,t}$ on $h$ and $ \{ h' \}$ stands for the dependence on the distribution of the random variables $h$, $\{h'\}$.\\

  Let us now proceed with the free-energy interpolation. First, from Eqs. (\ref{eq25}), (\ref{eq22}), (\ref{eq9})  it  is easy to show that
  \be \label{24}
  \phi_{k+1,0}(h, \{ h' \}) = \phi_{k,1} (h , \{ h'_a + \beta 2^{(k+1)(1/2-\sigma)} \sqrt{q_a - q_{a-1}} J^{k+1}_a \} ).
  \ee
  The derivative of  $\phi_{k+1,t}$ with respect to $t$ can be computed with a step-by-step repetition of the RSB-bound proof for the SK model \cite{guerra2003broken}. Given the average $\omega$ associated with the Boltzmannfaktor  (\ref{eq22}) and the respective replicated average $\Omega$, we define the averages $\tilde{\omega}_0, \cdots, \tilde{\omega}_{K}$ and the respective replicated averages $\tilde{\Omega}_0, \cdots, \tilde{\Omega}_{K}$ \cite{guerra2003broken}  as
  \beas
  \tilde{\omega}_K(\cdot) \equiv \omega(\cdot), \; \tilde{\omega}_a(\cdot) \equiv \mathbb{E}_{a+1} \cdots \mathbb{E}_{K}[f_{a+1} \cdots f_{K} \omega (\cdot)].
  \eeas
  Setting
  \be
  f_a \equiv \frac{Z^{k+1}_{a,t}(h, \{ h' \})^{m_a}}{\mathbb{E}_a[Z^{k+1}_{a,t}(h, \{ h' \})^{m_a}]}
  \ee
  for $a=1,\cdots, K$, and $\langle \cdot \rangle_a \equiv  \mathbb{E}[f_1 \cdots f_a  \tilde{\Omega}_a(\cdot)]$ for $a=0, \cdots, K$, we obtain
  \bea \label{eq23}
  \frac{d\phi_{k+1,t}(h, \{ h' \})}{dt}&  =&    \frac{\beta^2}{4} 2^{(k+1)(1-2 \sigma)} \left[ \sum_{a=0}^{K}  (m_{a+1}-m_a) q_a^2 -1 \right] + \\ \nn
  &&   - \frac{\beta^2}{4} 2^{(k+1)(1-2 \sigma)} \sum_{a=0}^{K} (m_{a+1}-m_a) \langle (R_{12} - q_a)^2 \rangle_a.
  \eea
  Using Eqs. (\ref{24}), (\ref{eq23}) we obtain the recursive inequality
  \bea \label{eq24}
  \phi_{k+1,1}(h, \{ h' \}) &= &\phi_{k+1,0}(h, \{ h' \}) + \int_0^1 \frac{d\phi_{k+1,t}(h, \{ h' \})}{dt}  dt \\  \nn
  &\leq & \phi_{k,1} (h , \{ h'_a + \beta 2^{(k+1)(1/2-\sigma)} \sqrt{q_a - q_{a-1}} J^{k+1}_a \} ) +\\ \nn
  & & +   \frac{\beta^2}{4} 2^{(k+1)(1-2 \sigma)} \left[ \sum_{a=0}^{K}  (m_{a+1}-m_a) q_a^2 -  1 \right].
  \eea
  From Eqs. (\ref{eq25}), (\ref{eq22}), (\ref{eq9}), it is easy to show that
  \be\label{eq28}
  \phi_{k+1,1}(h, \vec{0}) = f_{k+1}.
  \ee
  By using Eq. (\ref{eq28}) and iterating Eq. (\ref{eq24}) for $k+1, k, \cdots, 1$, we obtain
  \bea \label{eq30}
  \phi_{k+1,1}(h, \{ h' \}) &\leq & \phi_{k,1}  (h , \{ h'_a + \beta 2^{(k+1)(1/2-\sigma)} \sqrt{q_a - q_{a-1}} J^{k+1}_a \} ) +\\ \nn
  & & +   \frac{\beta^2}{4} 2^{(k+1)(1-2 \sigma)}\left[ \sum_{a=0}^{K}  (m_{a+1}-m_a) q_a^2 - 1 \right] \\  \nn
  & \leq & \cdots \\ \nn
  & \leq & \phi_{1,0} \left(h ,\left\{ \beta \sum_{l=2}^{k+1}2^{l(1/2-\sigma)} \sqrt{q_a - q_{a-1}} J^{l}_a \right\} \right) + \\ \nn
  && +    \frac{\beta^2}{4} \sum_{l=1}^{k+1} 2^{l(1-2 \sigma)} \left[\sum_{a=0}^{K}  (m_{a+1}-m_a) q_a^2 -  1 \right].
  \eea
  From the definition of the interpolating Hamiltonian, Eq. (\ref{eq25}), it is easy to show that the first term in the last line of Eq. (\ref{eq30}) is given by the free energy of a single-spin system, and that this is equal to $\log 2 + \mathbb{E}[\log Z_0]$, where $Z_0$ is defined by Eq. (\ref{eq26}).
\end{proof}

The bound of Theorem \ref{thm3}, depends on the parameters 
\be\label{eq41}
q_1, \cdots, q_{K-1}, \, m_1, \cdots, m_K.
\ee
 By minimizing the right-hand side of Eq. (\ref{eq38}) with respect to these parameters, one obtains the best estimate of the free energy according to this RSB  bound. It is important to point out that the bound (\ref{eq38}) can be generalized by letting the parameters (\ref{eq41}) depend on the hierarchical level: 
\be\label{eq42}
\{ q_1^l, \cdots, q_{K-1}^l, \, m_1^l, \cdots, m_K^l\}_l,
\ee
where $l=1, \cdots, k+1$. It is easy to check that the parameter values realizing the minimum of such bound are level-independent 
\bea
q_a^1 = q_a^2 = \cdots = q_a^{k+1}, \, a=1, \cdots, K-1,\\
m_a^1 = m_a^2 = \cdots = m_a^{k+1}, \, a=1, \cdots, K.
\eea
Hence, in Theorem \ref{thm3} we considered directly the case where the bound parameters are independent of the hierarchical level. \\

Theorem \ref{thm3} establishes a RSB bound for the free energy of the HEA. It is easy to show that this bound is based on a MF picture: Since the bound is obtained as a recursive iteration of Eq. (\ref{eq24}), the bound reminder is given by a sum over all levels $l=1,\cdots,k+1$ of the last term in Eq. (\ref{eq23}), which represents the fluctuations of the order parameter $R_{12}$ within a block of $2^l$ spins with respect to the values $q_0, \cdots, q_{K}$. Since in Definition  \ref{eq10} of the HEA the interaction at the $l$-th level is a MF one, for large $l$  we expect these blocks to have a MF-like behavior, i.e. we expect  OPF to be suppressed. Differently, for small $l$ the fluctuations of $R_{12}$ are not small, and neither is the reminder in Eq. (\ref{eq23}). It follows that in order to improve upon the MF bound of Theorem \ref{thm3}, we should account for the OPF arising in small blocks of spins. In what follows, we propose a new scheme to account for these fluctuations that fully exploits the hierarchical structure of the model. In particular, in the next Section we illustrate this idea for DHM, and show that this new scheme accounts for OPF, yielding a free-energy bound that improves upon the MF one.

\section{Dyson's Hierarchical Model}

Dyson's hierarchical model  is a system of $2^{k+1}$ Ising spins $S_i = \pm 1$ labeled by index $i=1,2,\cdots, 2^{k+1}$, whose Hamiltonian $H_{k+1}[\vec{S}]$ is introduced recursively by the following

\begin{defn}\label{def2}
  The Hamiltonian of DHM is defined by
  \be
  H_{k+1}[\vec{S}]=H_{k}[\vec{S}_{1}]+H_{k}[\vec{S}_{2}]-\frac{J}{2^{(k+1)2\sigma}}\sum_{i<j=1}^{2^{k+1}}S_{i}S_{j}
  \ee
where $\vec{S}_1 \equiv \{S_i\}_{1 \leq i \leq 2^k}$,  $\vec{S}_2 \equiv \{S_i\}_{2^k+1 \leq i \leq 2^{k+1}}$, 
 $H_{0}[S]=0$, $J\geq0$ and $\sigma$ is a number.
\end{defn}

Like for the HEA, the number $\sigma$ in Definition \ref{def2} determines how fast the spin-spin interactions decrease with distance.\\

The existence of the thermodynamic limit for the free energy of DHM
\be
f_{k+1}\equiv\frac{1}{2^{k+1}}\log\sum_{\vec{S}}\exp\left(-\beta H_{k+1}[\vec{S}]+h\sum_{i=1}^{2^{k+1}}S_{i}\right),
\ee
has been proven  by Gallavotti and Miracle-Sole \cite{gallavotti1967statistical}. Here, we first prove the analogous of the MF bound, Theorem \ref{thm3}, previously derived for the HEA model.

\begin{thm}[Mean-field bound] \label{thm:1}
  Given  $-1\leq m \leq1$, one has
  \bea \label{eq:9}
  f_{k+1} & \geq & \log2+\log\cosh\left[\beta J\sum_{l=1}^{k+1}2^{l(1-2\sigma)}m+h\right] + \\ \nn
  && -\frac{\beta J}{2}\left[\sum_{l=1}^{k+1} 2^{l(1-2\sigma)}m^{2}+\sum_{l=1}^{k+1}\frac{1}{2^{2l\sigma}}\right] \\ \nn
  & \equiv & \phi_{k+1}^{ \textrm{MF}}(m),
  \eea
  where $\textrm{MF}$ stands for mean field.
\end{thm}
\begin{proof}
  Let us define the interpolating Hamiltonian $H_{k+1,t}[\vec{S}]$, the associated partition function $Z_{k+1,t}(h)$  and the free energy $\phi_{k+1,t}(h)$ as
  \bea\label{eq31}
  H_{k+1,t}[\vec{S}]  &\equiv&  - \frac{Jt}{2^{2(k+1)\sigma}}\sum_{i>j=1}^{2^{k+1}}S_{i}S_{j} -(1-t)m J 2^{(k+1)(1-2\sigma)}\sum_{i=1}^{2^{k+1}}S_{i} +\\  \nn
  &&+ H_{k}[\vec{S}_{1}]+H_{k}[\vec{S}_{2}],\\
  \label{eq32}
  Z_{k+1,t}(h) & \equiv & \sum_{\vec{S}}\exp\left(-\beta H_{k+1,t}[\vec{S}]+h\sum_{i=1}^{2^{k+1}}S_{i}\right),\\
  \label{eq33}
  \phi_{k+1,t}(h) & \equiv & \frac{1}{2^{k+1}}\log Z_{k+1,t}(h).
  \eea
  Using Eqs.  (\ref{eq31}), (\ref{eq32}),  (\ref{eq33}), it is easy to show that
  \begin{eqnarray}\label{eq36}
    \phi_{k+1,1}(h) &=& f_{k+1},\\
    \label{eq34}
    \phi_{k+1,0}(h) &=& \phi_{k,1} (h+\beta J m 2^{(k+1)(1-2\sigma)} ),\\
    \label{eq35}
    \frac{d\phi_{k+1,t}(x,h)}{dt} & = & -\frac{\beta J}{2}\left( 2^{(k+1)(1-2\sigma)} m^{2}+2^{-2(k+1)\sigma}\right) +\\ \nn
    &  & +\frac{\beta J}{2} 2^{(k+1)(1-2\sigma)} \langle (M-m)^2\rangle_{t},
  \end{eqnarray}
  where $M\equiv \frac{1}{2^{k+1}} \sum_{i=1}^{2^{k+1}} S_i$ is the magnetization within a block of $2^{k+1}$ spins and  $\langle \cdot \rangle_t$ stands for the average associated with the Boltzmannfaktor (\ref{eq31}).\\

  From Eqs. (\ref{eq36}), (\ref{eq34}), (\ref{eq35}) we have
  \bea\label{eq37}
  \phi_{k+1,1}(h) & = & \phi_{k+1,0}(h) + \int_0^1 \frac{d\phi_{l,t}(x,h)}{dt} dt \\ \nn
  & \geq & \phi_{k,1} (h+\beta J m 2^{(k+1)(1-2\sigma)} )   -\frac{\beta J}{2}\left( 2^{(k+1)(1-2\sigma)} m^{2}+2^{-2(k+1)\sigma}\right) \\ \nn
  & \geq & \cdots \\ \nn
  &\geq&  \phi_{1,0}\left(h+\beta J m \sum_{l=2}^{k+1}2^{l(1-2\sigma)}\right)  - \frac{\beta J}{2}\left( \sum_{l=1}^{k+1} 2^{l(1-2\sigma)} m^{2}+\sum_{l=1}^{k+1}2^{-2l\sigma}\right),
  \eea
  where in Eq. (\ref{eq37}) we have recursively  used Eq. (\ref{eq34}) for $k+1, k, \cdots, 1$.  Using Eqs. (\ref{eq31}), (\ref{eq32}), (\ref{eq33}), it is easy to show that the last line in Eq. (\ref{eq37}) implies Eq. (\ref{eq:9}).
\end{proof}

A direct inspection of the reminders in bounds (\ref{eq38}), (\ref{eq:9})--Eqs. (\ref{eq23}) and (\ref{eq35}) respectively--shows that the  bounds in  Theorems \ref{thm3} and  \ref{thm:2} are exact only if OPF vanish, as one would expect in a MF scenario. Here, we propose a novel method providing a NMF bound that accounts for non-vanishing OPF. The method is described in the following

\begin{thm}[Non-mean-field bound]\label{thm:2}
  Given $-1 \leq m \leq 1$, one has
  \begin{eqnarray}\label{eq:4}\nn
    f_{k+1} & \geq & \log 2 + \log  \cosh\left[ \beta J \left( \sum_{l=1}^{k+1}2^{l(1-2\sigma)} -  \sum_{l=1}^{k+1}2^{-2 l \sigma}  \right) m+h\right] + \\
    && - \frac{\beta J }{2} \left( \sum_{l=1}^{k+1}2^{l(1-2\sigma)} -  \sum_{l=1}^{k+1}2^{-2 l \sigma}   \right) m^2\\ \nn
    & \equiv & \phi_{k+1}^{\textrm{NMF}}(m),
  \end{eqnarray}
  where $\textrm{NMF}$ stands for non mean field.
\end{thm}
\begin{proof}
  Let us take $x\geq0$, $0 \leq t \leq 1$ and let us introduce the interpolating Hamiltonian
  \be\label{eq13}
  H_{k+1,t}[\vec{S}]  \equiv  -t\, u[\vec{S}]-(1-t)\, v[\vec{S}]+H_{k}[\vec{S}_{1}]+H_{k}[\vec{S}_{2}],
  \ee
  with
  \bea \label{eq11}
  u[\vec{S]} & \equiv & \frac{J}{2^{2(k+1)\sigma}}\sum_{i>j=1}^{2^{k+1}}S_{i}S_{j}+\frac{xJ}{2 \cdot 2^{2(k+1)\sigma}}\sum_{i,j=1}^{2^{k+1}}\left(S_{i}-m\right)\left(S_{j}-m\right),\\\nn
  v[\vec{S}] & \equiv & \frac{J(1+x)}{2\cdot 2^{2(k+1)\sigma}}\left[\sum_{i,j=1}^{2^k}\left(S_{i}-m\right)\left(S_{j}-m\right)+\sum_{i,j=2^k+1}^{2^{k+1}}\left(S_{i}-m\right)\left(S_{j}-m\right)\right]+\\
  &  & + m J 2^{(k+1)(1-2\sigma)}\sum_{i=1}^{2^{k+1}}S_{i}. \label{eq12}
  \eea
  The partition function and free energy associated with the Hamiltonian (\ref{eq13})  are
  \bea \label{eq14}
  Z_{k+1,t}(x,h) & \equiv & \sum_{\vec{S}}\exp\left(-\beta H_{k+1,t}[\vec{S}]+h\sum_{i=1}^{2^{k+1}}S_{i}\right),\\
  \label{eq15}
  \phi_{k+1,t}(x,h) & \equiv & \frac{1}{2^{k+1}}\log Z_{k+1,t}(x,h).
  \eea

  Let us proceed with the interpolation: First, from Eqs. (\ref{eq13}),   (\ref{eq11}),  (\ref{eq12}), (\ref{eq14}),  (\ref{eq15}),  we relate $\phi_{k+1,0}$ to $\phi_{k,1}$
  \bea\label{eq19}
  \phi_{k+1,0}(x,h) = \phi_{k,1}\left(\frac{1+x}{2^{2\sigma}},h+\beta J m 2^{(k+1)(1-2\sigma)}\right).
  \eea
  Using the same definitions as above, it is easy to show that the derivative of $\phi_{k+1,t}$ with respect to $t$ reads
  \begin{eqnarray} \nn
    \frac{d\phi_{k+1,t}(x,h)}{dt} & = & -\frac{\beta J}{2}\left( 2^{(k+1)(1-2\sigma)} m^{2}+2^{-2(k+1)\sigma}\right) +\label{eq:1}\\
    &  & +\frac{\beta J(1+x)}{2^{(k+1)(1+2\sigma)}}\sum_{2^{k}+1\leq i\leq2^k}\,\,\sum_{1\leq j\leq2^k}\langle\left(S_{i}-m \right)\left(S_{j}-m\right)\rangle_{t},
  \end{eqnarray}
  where $\langle \cdot \rangle_t$ denotes the average associated with the Boltzmannfaktor (\ref{eq14}). \\

  It is easy to show that each term in the sum in Eq. (\ref{eq:1}) is non-negative
  \be \label{eq16}
  \langle\left(S_{i}-m \right)\left(S_{j}-m\right)\rangle_{t} \geq 0.
  \ee
  Indeed, because of the translational invariance of the Hamiltonian $H_{k+1,t}$,  the average $\langle S_i \rangle_t$ does not depend on the lattice site $i$. Hence, the LHS of Eq. (\ref{eq16}) reads
  \be\label{eq17}
  \langle\left(S_{i}-m \right)\left(S_{j}-m\right)\rangle_{t} = \langle S_i S_j \rangle_t - 2 m \langle S_i \rangle_t + m^2.
  \ee
  Since  $H_{k+1,t}$ is a ferromagnetic Hamiltonian, Griffith's inequalities for the connected correlation functions \cite{griffiths1967correlationsIII} hold
  \be\label{eq18}
  \langle S_i S_j \rangle_t  - \langle S_i\rangle_t \langle S_j\rangle_t \geq 0.
  \ee
  Putting together Eqs. (\ref{eq17}), (\ref{eq18}), we obtain Eq. (\ref{eq16})
  \be
  \langle\left(S_{i}-m \right)\left(S_{j}-m\right)\rangle_{t} \geq (\langle S_i \rangle_t - m)^2 \geq 0.
  \ee
  Thus, Eqs. (\ref{eq19}), (\ref{eq:1}) and (\ref{eq16}) imply
  \bea\label{eq20}
  \phi_{k+1,1}(x,h) & = & \phi_{k+1,0}(x,h) + \int_0^1 \frac{d\phi_{l,t}(x,h)}{dt} dt \\ \nn
  & \geq & \phi_{k,1}\left(\frac{1+x}{2^{2\sigma}},h+\beta J m 2^{(k+1)(1-2\sigma)}\right)  + \\ \nn
  &&  -\frac{\beta J}{2}\left( 2^{(k+1)(1-2\sigma)} m^{2}+2^{-2(k+1)\sigma}\right) .
  \eea
  Equation (\ref{eq20}) is a recursive inequality relating $\phi_{k+1,1}$ to $\phi_{k,1}$: To obtain a bound for the free energy $f_{k+1}$, we notice that $\phi_{k+1,1}(0,h) = f_{k+1}$ and--proceeding  as in Theorem \ref{thm1}--we exploit the hierarchical structure of the model  by iterating  recursively Eq. (\ref{eq20}) until the level $k=1$ is reached:
  \bea\label{eq21}
  f_{k+1} & = & \phi_{k+1}(0,h) \\ \nn
  & \geq & \phi_{k,1}\left(\frac{1}{2^{2\sigma}},h+\beta J m 2^{(k+1)(1-2\sigma)}\right)    -\frac{\beta J}{2}\left( 2^{(k+1)(1-2\sigma)} m^{2}+2^{-2(k+1)\sigma}\right)\\\nn
  & \geq & \cdots \\ \nn
  & \geq & \phi_{1,0}\left( \sum_{l=1}^{k} 2^{-2l\sigma}  ,h+\beta J m \sum_{l=2}^{k+1}2^{l(1-2\sigma)}\right)  +  \frac{\beta J}{2}\Bigg( \sum_{l=1}^{k+1} 2^{l(1-2\sigma)} m^{2} + \\ \nn
  && +  \sum_{l=1}^{k+1}2^{-2l\sigma}\Bigg).
  \eea
  By using  again Eqs. (\ref{eq13}), (\ref{eq11}),  (\ref{eq12}), (\ref{eq14}), (\ref{eq15}), Eq. (\ref{eq21}) leads to Eq. (\ref{eq:4}).
\end{proof}
Let us now compare the MF bound, Theorem \ref{thm:1}, with the NMF bound, Theorem \ref{thm:2}. In Theorem \ref{thm:1} the bound reminder is given by the OPF $\langle (M-m)^2 \rangle_t$. By rewriting the magnetization $M$ in terms of the magnetizations in the left and right blocks of $2^k$ spins $M_1 \equiv \sum_{i=1}^{2^k}S_i$, $M_2 \equiv \sum_{i=2^k+1}^{2^{k+1}}S_i$--namely $M=\frac{M_1+M_2}{2}$--we can write this reminder as $\frac{1}{2}\langle (M_1-m)^2 + (M_2-m)^2 \rangle_t + \langle (M_1-m)(M_2-m)\rangle_t$. In Theorem \ref{thm:2} the bound reminder is given only by $\langle (M_1-m)(M_2-m)\rangle_t$: The OPF within the left and right block--$\langle (M_1-m)^2\rangle_t$ and $\langle(M_2-m)^2 \rangle_t$ respectively--have been reabsorbed into an effective Hamiltonian of the left and right block, i.e. the term in brackets in Eq. (\ref{eq12}).  Hence, we expect the bound of Theorem \ref{thm:2} to improve upon the bound of Theorem \ref{thm:1}.  We explicitly show this in Fig. \ref{fig1}, where we plot the thermodynamic limit of the MF and NMF bound
\beas
\phi_{\textrm{MF}}(m) & \equiv & \lim_{k \rightarrow \infty} \phi_{k+1}^{\textrm{MF}}(m),\\
\phi_{\textrm{NMF}}(m) & \equiv & \lim_{k \rightarrow \infty} \phi_{k+1}^{\textrm{NMF}}(m),
\eeas
for a given value of $1/2 < \sigma \leq 1$, $\beta$ and $h=0$, and we show that
\[
\max_{m \in [-1,1]}\phi_{\textrm{MF}}(m) < \max_{m \in [-1,1]} \phi_{\textrm{NMF}}(m).
\]\\

\begin{figure}[h]
  \centering\includegraphics[scale=1.2]{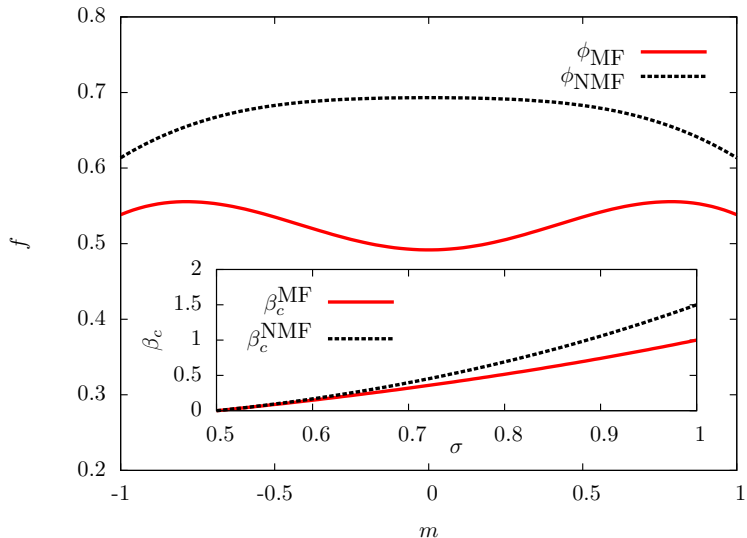}
  \caption{Mean-field  bound $\phi_{\textrm{MF}}$ and non-mean-field bound $\phi_{\textrm{NMF}}$ for the free energy of DHM as functions of $m$ for, $\beta =1$, $\sigma = 0.9$ and  $h=0$. Since $\beta_{c}^{\textrm{MF}} <\beta < \beta_{c}^{\textrm{NMF}}$, the mean-field bound is in the low-temperature phase ($\arg\max_{m \in [-1,1]}\phi_{\textrm{MF}}(m) \neq 0$), while the non-mean-field bound is in the high-temperature phase ($\arg\max_{m \in [-1,1]}\phi_{\textrm{NMF}}(m) = 0$). Inset: Inverse critical temperatures $\beta_{c}^{\textrm{MF}}$ and $\beta_{c}^{\textrm{NMF}}$ of the mean-field and non-mean-field bound respectively for  $h=0$ as functions of $1/2 < \sigma \leq 1$. \label{fig1}}
\end{figure}

It is easy to show that for both bounds there is a critical value $\beta_c$ of the inverse temperature $\beta$ such that the maximum of $\phi$ is realized for $m = 0 $ if $\beta \leq \beta_c$, while the maximum is realized for $m \neq 0$ if $\beta > \beta_c$. At this value of the inverse temperature, a ferromagnetic phase transition takes place  \cite{dyson1969existence}. From Eqs. (\ref{eq:9}), (\ref{eq:4}) it is straightforward to show that the inverse critical temperatures associated with  $\phi_{\textrm{MF}}$ and $\phi_{\textrm{NMF}}$ are  $\beta_{c}^{\textrm{MF}} = 2^{2\sigma-1}-1$ and $\beta_{c}^{\textrm{NMF}} =2^{1-2\sigma} -3+2^{2\sigma}$ respectively: These inverse critical temperatures are depicted in the inset of Fig. \ref{fig1} as functions of $\sigma$ in the interval $1/2 < \sigma \leq 1$ where the thermodynamic limit of DHM is well defined and where a finite-temperature phase transition is known to occur in the model \cite{dyson1969existence}. Given that the NMF bound (\ref{eq:9}) treats the spin-spin interactions between left and right blocks differently from the interactions within blocks, this bound accounts for a spatial structure in spin-spin couplings, in particular for the decrease of the interaction strength with distance. Differently, in the MF bound (\ref{eq:4}) inter-block and intra-block interactions are treated in the same way, and there is no hallmark of a spatial structure. Compared to a system with infinite-range couplings, a system whose interactions decrease with distance needs to be cooled down to lower temperatures to enter into the ordered phase: Hence, we expect the inverse critical temperature of the NMF bound to be smaller than that of the MF bound  \cite{griffiths1967correlationsIII}, as shown in the inset of Fig. \ref{fig1}.

\section{Conclusions and Outlook}

In this paper we studied two non-mean-field spin models built on a hierarchical lattice, the hierarchical Edwards-Anderson model (HEA) \cite{franz2009overlap} of a spin glass and Dyson's hierarchical model (DHM) \cite{dyson1969existence} of a ferromagnet. For the HEA, we proved the existence and self-averaging of the free energy in the thermodynamic limit. In addition, we have extended  to the HEA the mean-field (MF) replica-symmetry-breaking (RSB) bounds for the free energy first derived for the MF Sherrington-Kirkpatrick model of a spin glass.  We have then proposed a novel method to improve upon these MF bounds. We have applied this method to DHM, and we have shown that it provides a tighter free-energy bound compared to the MF one, and a value of the critical temperature closer to the exact one. To extend our method to the HEA, one needs to extend Griffith's correlation inequalities for Ising ferromagnets \cite{griffiths1967correlationsII} to hierarchical spin glasses, which we leave as a topic of future studies.

\section*{Acknowledgments}
M. C. is grateful to S. Franz for useful discussions, to NSF for funding through Grants PHY--0957573 and CCF--0939370, to the Human Frontiers Science Program, to the Swartz Foundation, and to the W. M. Keck Foundation for financial support.\\
 A. B. is grateful to MIUR for funding trough the grant FIRB RBFR08EKEV, and to Sapienza Universit\`{a} di Roma and to GNFM-INdAM for partial financial support.\\ F. G. is grateful to Sapienza Universit\`{a} di Roma and to INFN for partial financial support.

\bibliographystyle{unsrt}
\bibliography{bibliography}

\end{document}